\documentclass[12 pt]{article}

\usepackage{graphicx}
\usepackage{amsthm}
\usepackage{amssymb, amsmath}
\usepackage{mathtools}
\usepackage{color}
\usepackage{soul,xcolor}
\usepackage{url}
\usepackage{tensor}
\usepackage{dsfont}
\usepackage{cite}
\usepackage{matlab-prettifier}
\usepackage{wrapfig}
\usepackage{subcaption}
\usepackage[export]{adjustbox}
\begin{document}
\title{Quantum information recast via multiresolution  in $L_2(0,1]$}
\author{Mandana Bidarvand \& Artur Sowa \\
Department of Mathematics and Statistics\\
University of Saskatchewan \\
106 Wiggins Road,
Saskatoon, SK S7N 5E6 \\
Canada  }

\date{}
\maketitle
\newtheorem{definition}{Definition}[section]
\newtheorem{theorem}{Theorem}[section]
\newtheorem{proposition}{Proposition}[section]
\newtheorem{lemma}{Lemma}[section]
\newtheorem{remark}[theorem]{Remark}
\newtheorem{corollary}{Corollary}[section]
\newtheorem{algorithm}{Algorithm}[section]

\begin{center}

 Abstract

 \end{center}
We present a multiresolution approach to the theory of quantum information. It arose from an effort to develop a systematic mathematical approach to the analysis of an infinite array of qubits, i.e., a structure that may be interpreted as a quantum metamaterial. Foundational to our approach are two mathematical constructions with classical roots: the Borel isomorphism and the Haar basis. Here, these constructions are intertwined to establish an identification between $L_2(0,1]$ and the Hilbert space of an infinite array of qubits and to enable analysis of operators that act on arrays of qubits (either finite or infinite). The fusion of these two concepts empowers us to represent quantum operations and observables through geometric operators. As an unexpected upshot, we observe that the fundamental concept of calculus is inherent in an infinite array of qubits; indeed, the antiderivative arises as a natural and indispensable operator in this context. 


 \section{Introduction} \label{Introduction}
 
A quantum metamaterial is an engineered structure whose physical properties and modes of interaction with the environment depend on its quantum state. A structure consisting of an array of qubits that is enabled to interact with the electromagnetic field is an example of such a material, and has been initially discussed in physics literature, e.g., \cite{Zagoskin2}, \cite{Zagoskin3}, \cite{Rakhmanov}. A mathematical approach has been subsequently proposed in \cite{R1}. The underlying physical model is the generalized Jaynes-Cummings model where a mode of light is interacting with an array of $n$ qubits. The central part of the theory is the interaction Hamiltonian:
\begin{equation*}
\mathcal{H}_I = \hbar\, C_x^{(n)}\otimes (\hat{a} + \hat{a}^\dagger)
\end{equation*}
where of course, $\hat{a}$ and $\hat{a}^\dagger$ are the annihilation and creation operators of the harmonic oscillator, while
\begin{equation}
C_x^{(n)} = \sum_{k = 1}^{n}\,\lambda_k\, \sigma_x^{k},\quad \sigma_x^{k} = I\otimes \ldots I\otimes \sigma_x \otimes I \ldots
\,\otimes I \quad (n \mbox{ factors with } \sigma_x \mbox{ in k-th place}). \label{paulix}
\end{equation}
 The spectrum of the $2^n$-by-$2^n$ matrix $C_x^{(n)}$ can be described by a closed-form formula. Indeed, note that the Pauli matrix $\sigma_x$, defined in (\eqref{paulix}), is diagonalized via
\[
u \, \sigma_x u^\dagger = \sigma_z, \quad \mbox{ where } u = \frac{1}{\sqrt{2}}\left[\begin{array}{rr}
                        1 & 1 \\
                        1 & -1 
                      \end{array}\right].
\]
Thus,  $ u^{\otimes n}$ is a unitary matrix that diagonalizes $C_x^{(n)}$, namely
\begin{equation}\label{x-to-z}
     u^{\otimes n}\,C_x^{(n)}\, \left( u^{\otimes n}\right)^\dagger = C_z^{(n)} = \sum_{k = 1}^{n}\,\lambda_k\, \sigma_z^{k}. 
\end{equation}
The entries of the diagonal matrix $C_z^{(n)}$ are its eigenvalues and, of course the eigenvalues of $C_x^{(n)}$. They can be given explicitly via the formula
\begin{equation}\label{theEs}
  \vec{E} = \sum_{k = 1}^{n}\,\lambda_k\, R_k,\quad \mbox{ where } R_k = \left[\begin{array}{c}
                                                                                 1 \\
                                                                                 1 
                                                                               \end{array}\right]^{\otimes k-1} \otimes \left[
                                                                               \begin{array}{r}
                                                                                 1 \\
                                                                                 -1 
                                                                               \end{array}\right] \otimes \left[\begin{array}{c}
                                                                                       1 \\
                                                                                       1 
                                                                                     \end{array}\right]^{\otimes n-k}.
\end{equation}
Thus, each eigenvalue $\vec{E}(j), j = 1, 2, \ldots, 2^n$, is a linear combination of all parameters $\lambda_k$ with coefficients $\pm 1$. All possible selections of the sequences of $\pm 1$ are admissible, so there are $2^n$ eigenvalues some of which may coincide, depending on the values of $\lambda_k$.  
In other words the eigenvalues result from all combinations of the form
\begin{equation}\label{eigs_fin}
  \sum_{k=1}^{n} \epsilon_k \lambda_k, \quad \mbox{ where } \epsilon_k = \pm 1.
\end{equation}
It is natural to ask about the features of the eigenvectors of an operator such as $C_x^{(n)}$. The concept of scale is inherent and key in understanding the answer to this question. Specifically, we bring to bear the Haar transform, see e.g. \cite{Daubechies}. We summarize the construction of the Haar basis in Subsection \ref{sec_Haar_Borel}. The discrete version of the Haar transform is a finite matrix whose columns are discrete models of the Haar functions. By abuse of notation, we will denote it $\mathcal{T}_H$ and its size will be clear from the context. The basic observation is that 
\begin{equation}\label{Blocks}
 \mathcal{ T_H} \, C_x^{(n)} \, \mathcal{T_H}' \quad \mbox{ is a block matrix}.
\end{equation} 
The blocks lie along the diagonal and have dimensions $1$-by-$1$, another $1$-by-$1$, then $2$-by-$2$, $4$-by-$4$, $8$-by-$8$, etc.
Since the eigenvectors of $C_x^{(n)}$ correspond to the eigenvectors of consecutive blocks, they are characterized by scale. This is the phenomenon first addressed in \cite{R1}, and developed here. It is helpful to highlight the structure of blocks. We omit the details of the computation as the rigorous approach is developed in subsequent sections. Also, this discrete example is easy to reproduce via symbolic computation. Now, the first block in (\ref{Blocks}) is $1$-by$1$ with the entry
\[
\lambda_1 + \lambda_2 + \ldots \lambda_n.
\]
The corresponding eigenvector is a column of constants (the discrete version of the constant function). 
The second block is also $1$-by-$1$, as the corresponding eigenvector is the discrete version of the Haar function $H_{0,0}$, see Subsection \ref{sec_Haar_Borel}. The eigenvalue is
\[
-\lambda_1 + \lambda_2 + \ldots \lambda_n.
\]
Note that for the specific choice $\lambda_j = 2^{-j}$, the limit as $n\rightarrow\infty$ of this expression is zero. When $n$ is finite but large, this eigenvalue is close to $0$, but not exactly zero.
The third block is $2$-by-$2$; it has the structure
\[
\left(
  \begin{array}{cc}
    0 & \lambda_1 \\
    \lambda_1 & 0 \\
  \end{array}
\right) + (-\lambda_2 + \lambda_3 + \ldots + \lambda_n)\, I.
\]
(Throughout the article, we use $I$ to denote the identity matrix whose size is clear from the context.)
The corresponding eigenvectors are linear combinations of the (discrete versions of the) basis functions $H_{1,0}, H_{1,1}$.  Again, the diagonal part turns to zero for the specially scaled infinite array, while when $n$ is finite but large it will be close to zero. In the latter case, the eigenvalues of this block are close to those of the non-diagonal component, i.e.
\[
\pm \lambda_1
\]
The fourth block is $4$-by-$4$; it has the structure
\[
\left(
  \begin{array}{cccc}
    0 & \lambda_2 & \lambda_1 & 0 \\
    \lambda_2 & 0 & 0 & \lambda_1 \\
    \lambda_1 & 0 & 0 & \lambda_2 \\
    0 & \lambda_1 & \lambda_2 & 0 \\
  \end{array}
\right)
+ (-\lambda_3 + \lambda_4 + \ldots + \lambda_n)\, I.
\]
The corresponding eigenvectors are linear combinations of the (discrete) basis functions $H_{2,0}$, $H_{2,1}$, $H_{2,2}$, $H_{2,3}$.  As before, the diagonal part turns to zero for the specially scaled infinite array. Again, when $n$ is finite but large this may still be close to zero. The overall eigenvalues are close to those of the non-diagonal component, i.e.
\[
\pm \lambda_1 \pm \lambda_2
\]
This establishes the pattern. The blocks double in size at every next step. The nondiagonal part depends on one more coefficient $\lambda_j$, and the diagonal part is the tail that may be small. 
   
Using this observation as the starting point, it is possible to make sense of operators, such as $C_x = C_x^{(\infty)}$, $C_y = C_y^{(\infty)}$, and $C_z = C_z^{(\infty)}$, i.e., operators corresponding to an infinite array of qubits, provided one chooses the sequence of coefficients $\lambda_k$ judiciously. Note that the entries of $   u^{\otimes n}$ have magnitude of $2^{-n/2}$, and so this sequence of matrices does not possess a meaningful limit as $n\mapsto \infty$. In particular, in contrast to relation (\ref{x-to-z}), it cannot be expected that the resulting limit operators will remain unitarily equivalent. The limiting operators $C_x$, $C_z$ have been introduced in \cite{R1}. The procedure that yields a meaningful limit requires an assumption that $\lambda_k$'s scale as $2^{-k}$. It was established that the spectrum of both operators is precisely the interval $[-1,1]$. However, $C_z$ was found to be a multiplier, whereas $C_x$ was found to be an operator with a complete system of eigenfunctions, corresponding to eigenvalues that densely fill the interval $[-1,1]$. The main tool in the description of the limit itself has been the Haar transform. 

In this work we further develop the approach taken in \cite{R1}, and establish a rigorous foundation for the limit procedure. In addition to the Haar transform this calls for an explicit use of the Borel isomorphism, which establishes an identification of the Hilbert space of an array of qubits with the space of square-integrable functions on an interval. This setting allows the discussion of the case of infinite as well as finite arrays.  In fact, the multi-scale approach proposed here furnishes an alternative framework for discussion of quantum information in general. While it is entirely equivalent to the canonical one, it is advantageous in the analysis of certain types of problems, such as these addressed here.

As is well known, general classical simulations of quantum structures are limited to small systems, which is due to exponential dependence of the dimension of the underlying Hilbert space on the number of qubits. Also, the conventional theoretical methods of quantum many-body theory frequently involve only nearly factorized quantum states, failing to account for the truly quantum effects. In contrast the promise of quantum engineering is precisely in taking advantage of quantum correlations inherent in non-factorized states, \cite{Zagoskin}. This underscores the value of quantum models that can be solved explicitly, as these successfully addressed here.

 While the Haar transform is classical, \cite{Haar}, and well known in harmonic analysis and the Borel isomorphism is well known in the dynamical system theory, \cite{Anosov_etal}, their application in quantum theory appears to be new with the only precedent being the aforementioned reference \cite{R1}. The main result of this work is Theorem \ref{P+theor}. However, the main gain is the development of a systematic multi-scale approach to the theory and modelling of arrays of qubits. The most surprising observation is the fact that the antiderivative operator occupies a central position. It appears not to have been spotted in the quantum information arena until now.

\section{The space of quantum states of an infinite array of qubits} \label{Section_C_Haar}

In this section we discuss the space of states of an infinite array of qubits. It requires the Haar basis with its inherent structure of multiresolution, see, e.g., \cite{Wojtaszczyk}. It also requires an identification of quantum states with the so-called scaling functions, see Fig. \ref{Borel_identity}. The fundamental idea is that since the Hilbert space of a single qubit is $\mathbb{H}_Q = \mbox{span}\{|0\rangle, |1\rangle\}$, and the Hilbert space of an $n$-tuple of qubits is $ \bigotimes_{j=1}^n \mathbb{H}_Q$, the state space of an infinite array should be $ \bigotimes_{j=1}^\infty \mathbb{H}_Q$. Furthermore, one can interpret  $\mathbb{H}_Q $ as the space of functions on a two element set, say, $\mathbb{Z}_2$. Therefore, the infinite tensor product of such spaces should be a space of functions on the infinite Cartesian product $\mathbb{Z}_2 \times \mathbb{Z}_2\times \mathbb{Z}_2 \times \ldots $. However, this set can be identified  with the interval $(0,1]$ via the binary expansion. These ideas are related to a certain approach to dynamics and chaos theory, first considered by \'{E}. Borel, see \cite{Anosov_etal}. 

\subsection{The Haar basis and the Borel isomorphism}\label{sec_Haar_Borel}
First, recollect the structure of the Haar basis in $L_2(0,1]$, see e.g. \cite{Daubechies, Haar, Wojtaszczyk}. Let $G(x) \equiv 1$ for $x\in(0,1]$ and $G(x)=0$ everywhere else on the real line. We will use notation
\[G_{n,k}(x) = 2^{n/2} G(2^nx-k);
\]
in particular $G_{0,0} = G$. Note that $G_{n,k}$ is supported in the dyadic interval
\[\mbox{supp } G_{n,k} = I_{n,k} = (2^{-n}k,\, 2^{-n}(k+1)].
\]
Furthermore, let $H(x) = [G_{1,0}(x) -  G_{1,1}(x)]/\sqrt{2}$, and
\begin{equation}\label{the_Haar_fn}
 H_{n,k}(x) = 2^{n/2} H(2^nx-k).
\end{equation}
Note that $H_{0,0} = H$. It follows directly from definitions that
\begin{equation}\label{Gs2haar}
  H_{n,k} = \frac{1}{\sqrt{2}}\, \left( G_{n+1, 2k} - G_{n+1, 2k +1} \right)
\end{equation}
as well as 
\begin{equation}\label{Gs2G}
  G_{n,k} = \frac{1}{\sqrt{2}}\, \left( G_{n+1, 2k} + G_{n+1, 2k +1} \right).
\end{equation}
The following fundamental facts are  well known:
\begin{enumerate}
  \item
 Denote $V_n =\mbox{ span } \{G_{n,k}: k = 0, 1, 2,  \ldots 2^n-1 \}$. Then, $V_0 \subseteq V_1 \subseteq V_2 \ldots$, and
  \begin{equation}\label{multires}
       L_2(0,1] = \bigcup_{n=0}^\infty V_n \quad (\mbox{multiresolution ladder}).
  \end{equation}
  \item
  For $n = 0, 1, \ldots $ let $W_n = V_{n+1}/V_n$. Then $W_n = \mbox{ span } \{H_{n,k}: k = 0, 1, 2,  \ldots 2^n-1 \}$, so that
  \begin{equation}\label{direct_sum}
  L_2(0,1] = V_0\oplus \bigoplus_{n=0}^\infty W_n \quad (\mbox{direct sum decomposition}).
  \end{equation}
\item
In what follows we will make use of the orthogonal projections $\Pi_n$.
  The orthogonal projection can be defined as a map denoted as $\Pi_n: V_{n+1} \rightarrow W_n$, where $V_{n+1}$ represents the vector space of dimension $n+1$ and $W_n$ represents the vector space of dimension $n$. This map takes a vector from $V_{n+1}$ and projects it onto the subspace $W_n$, which is spanned by the Haar basis vectors corresponding to dimension $n$. This projection process results in a vector that lies entirely within the subspace $W_n$. Therefore, given a vector $x\in V_{n+1}$, in which $x$ is the sum of components in $V_n$ and $W_n$ (i.e. $x = v+w$), the orthogonal projection can be defined as:
  \begin{equation}
  \Pi (x)= w,
 \end{equation}
 where $w$ is the component of x that lies in $W_n$.

    \item
It follows that the set of functions $\{G_{0,0}\}\cup \{H_{n,k}:  k = 0, 1, 2,  \ldots 2^n-1; n = 0,1,2, \ldots\}$ furnishes an orthonormal basis in $L_2(0,1]$ (the Haar basis). In all references to this basis we will assume the canonical order in $W_n$ to be according to increasing $k$, so that, overall, the order of basis functions is fixed to be:
\[
G_{0,0}, H_{0,0}, H_{1,0}, H_{1,1}, H_{2,0}, H_{2,1}, H_{2,2}, H_{2,3}, \ldots
\]
 We let $\mathcal{T_H}: L_2(0,1] \rightarrow \ell_2$ denote the Haar transform which assigns to a square integrable function, say, $f$ its ordered sequence of Haar coefficients:
 \[
c_0 =\int_{0}^{1} f(x)\, dx,  \mbox{ and }  c_{n,k} = \int_{0}^{1} f(x) H_{n,k}(x)\, dx.
\]
Clearly, $\mathcal{T_H}$ is a unitary transformation.
\end{enumerate}

Next, let $\mathbb{H}_Q = \mbox{span}\{|0\rangle, |1\rangle\}$ be the Hilbert space of the qubit. The canonical basis in $ \bigotimes\limits_{j=1}^n \mathbb{H}_Q$ is given by vectors of the form
\[
| \epsilon_1 \, \epsilon_2\, \ldots \epsilon_{n} \rangle \quad \mbox { where each } \epsilon_k \in\{0,1\}.
\]
When the number of qubits is fixed to be $n$, it is often convenient to use the shorthand notation
\begin{equation}\label{short}
|\epsilon_1 \, \epsilon_2\, \ldots \epsilon_{n}\rangle  = |k\rangle \mbox{ where } k = 2^n \sum_{j=1}^{n} \frac{\epsilon_j}{2^j}, \mbox{ so that } k \in\{ 1,2,\ldots 2^n-1\}.
\end{equation}
However, this is generally ambiguous when the number of qubits is not fixed.
There is a one-to-one correspondence between sequences of bits and dyadic intervals. Namely,
\[
[ \epsilon_1 \, \epsilon_2\, \ldots \epsilon_{n}] \leftrightarrow I_{n,k} \quad \mbox{ where } k = 2^n \sum_{j=1}^{n} \frac{\epsilon_j}{2^j}.
\]
In other words, the sequence of bits is interpreted as the address of its corresponding interval.
 Moreover, for $n\geq 1$ consider the map
\begin{equation}\label{Borel_n}
B_n: \, \bigotimes\limits_{j=1}^n \mathbb{H}_Q \rightarrow V_n,
\end{equation}
given by
\begin{equation}\label{Borel_expl}
 B_n \,| \epsilon_1 \, \epsilon_2\, \ldots \epsilon_{n} \rangle = G_{n,k},\quad k = 2^n \sum_{j=1}^{n} \frac{\epsilon_j}{2^j}.
\end{equation}
This identification of qubits with the functions $G_{n,k}$ is illustrated in Fig. \ref{Borel_identity}.

\begin{figure}[!ht]
\centering
  \includegraphics[width=60mm]{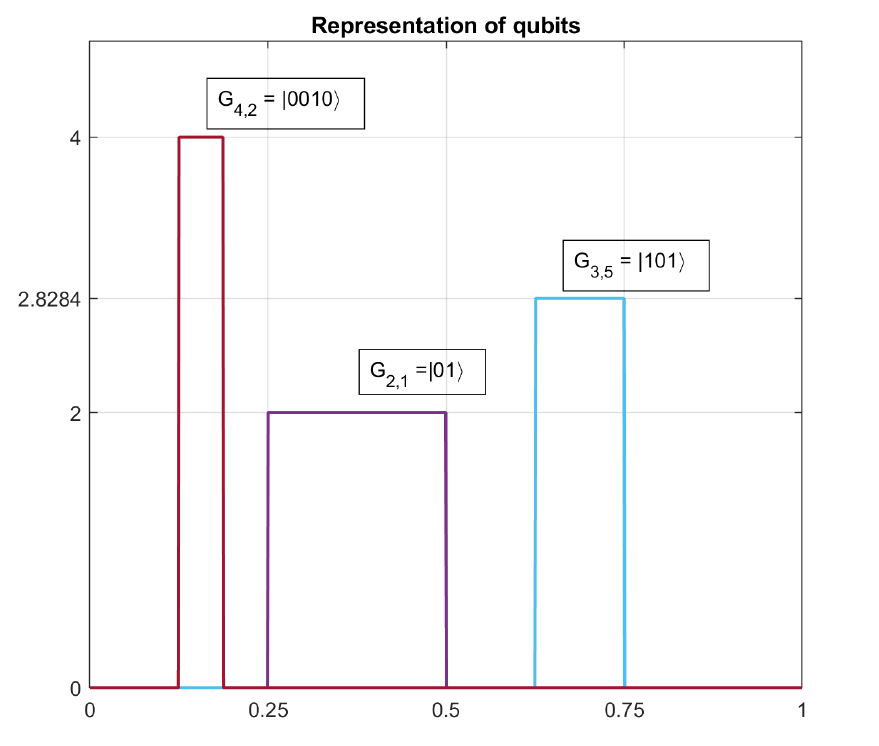}
\caption{\small Illustration of the identification given by \eqref{Borel_expl}, the symbol $B_n$ has been suppressed for clarity. }
\label{Borel_identity}
\end{figure}

When the number of qubits is fixed at $n$, this may be simplified to the useful form
\begin{equation}\label{Borel_short}
 B_n \,| k \rangle = G_{n,k} \quad \mbox{ for all }  k \in\{ 1,2,\ldots 2^n-1\}.
\end{equation}
Since $B_n$ takes one unitary basis to another, it is a unitary isomorphism. Note that since $V_n \subset V_{n+1}$, there are two ways of assigning qubits to vectors in $V_n$, i.e., one via $B_n$ and one via $B_{n+1}$. However, these are consistent, provided we use the following identification:
\begin{equation}\label{V_extend}
\iota_n : V_n \rightarrow V_{n+1}, \quad
  \iota_n [| \epsilon_1 \, \epsilon_2\, \ldots \epsilon_{n} \rangle]= \frac{1}{\sqrt{2}}\left(| \epsilon_1 \, \epsilon_2\, \ldots \epsilon_{n}\, 0 \rangle
 + | \epsilon_1 \, \epsilon_2\, \ldots \epsilon_{n}\, 1 \rangle \right).
\end{equation}
In other words, $n$-qubits state is identified with an $n+1$-qubit state with a completely unresolved last qubit.  Relations \eqref{V_extend} and \eqref{Gs2G} imply that the following diagram commutes
\begin{equation}\label{diagram_cons}
\begin{array}{lll}
  V_n & \hookrightarrow & V_{n+1}\\
  & \\
  \updownarrow B_n & & \updownarrow B_{n+1} \\
  & \\
  \bigotimes_{j = 1}^n \mathbb{H}_Q & \xrightarrow{\iota_n} &  \bigotimes_{j = 1}^{n+1} \mathbb{H}_Q
\end{array}
\end{equation} 
In this way, the collection of $1$-qubit, $2$-qubit, etc., including $N$-qubit states, can be identified with the collection of $N+1$-qubits states, via the identification given by consecutive application of $\iota_n$. This furnishes an isomorphism:
\begin{equation}\label{two_equiv}
  \mathbb{C}\oplus\bigoplus\limits_{n = 1}^N \bigotimes\limits_{j=1}^n \mathbb{H}_Q \, \equiv \, \bigotimes\limits_{j=1}^{N+1} \mathbb{H}_Q.
\end{equation}
The fist component $\mathbb{C}$ stands for the $0$-qubit states. It is necessary to complete the isomorphism. 
Passing to the limit $N\rightarrow \infty$, we obtain a definition of the state space of an infinite array of qubits. The left-hand side furnishes a rigorous definition of the right-hand side, i.e., $\bigotimes_{j=1}^{\infty} \mathbb{H}_Q$. We will often use the latter symbol with the understanding that its meaning is furnished by the identity
\begin{equation}\label{two_equiv_inf}
  \mathbb{C}\oplus\bigoplus\limits_{n = 1}^\infty \bigotimes\limits_{j=1}^n \mathbb{H}_Q \, \equiv \, \bigotimes\limits_{j=1}^{\infty} \mathbb{H}_Q.
\end{equation}
Note that the left-hand side has the form of a Fock space. 
We also observe the form of the isomorphism on $W_n = V_{n+1}\ominus V_n$; namely, it follows from \eqref{Gs2haar} that
\begin{equation}\label{Haar2qubits}
\begin{array}{cc}
   B_{n+1}^{-1} [H_{n,k}] & =  \frac{1}{\sqrt{2}}\, \left( B_{n+1}^{-1} [ G_{n+1, 2k}] -  B_{n+1}^{-1} [G_{n+1, 2k +1}] \right) \\
   & \\
 & =  \frac{1}{\sqrt{2}}\, \left(|\,\epsilon_1 \, \epsilon_2\, \ldots \epsilon_{n}\, 0\rangle -
|\,\epsilon_1 \, \epsilon_2\, \ldots \epsilon_{n}\, 1\rangle \right).
\end{array}
\end{equation}
where $k = 2^n \sum_{j=1}^{n} \, \epsilon_j/2^j$. 
In light of (\ref{multires}), one obtains the complete Borel isomorphism
\begin{equation}\label{Borel}
  B:\, \mathbb{C}\oplus\bigoplus\limits_{n = 1}^\infty \bigotimes\limits_{j=1}^n \mathbb{H}_Q \rightarrow  V_0\oplus \bigoplus_{n=0}^\infty W_n = L_2(0,1]  ,
\end{equation}
where we set $B|_\mathbb{C} [1] = G_{0,0}$.
The unitary map $B$ furnishes an identification o the space of states of an infinite array of qubits with square integrable functions on the interval $(0,1]$. 

\subsection{Quantum operations}

Quantum operations on states are unitary maps
\[
U: \bigotimes\limits_{j=1}^{\infty} \mathbb{H}_Q\rightarrow \bigotimes\limits_{j=1}^{\infty} \mathbb{H}_Q
\quad \mbox{ equivalently } \quad    B\, U \, B^{-1}:\, L_2(0,1]\rightarrow L_2(0,1].
\]
For simplicity, we will write $U$ instead of $B\, U \, B^{-1}$ when the meaning of the map is clear from context. Similarly, we will use the term state and function interchangeably.
 
At first, we will consider operators that only affect the state of the first $n$ qubits. We will call such operators \emph{local}. Invoking the Borel isomorphism we can state that such an operator is fully determined via its action in the space $V_n$. However, we need to examine how this action extends to the entire space $L_2(0,1]$. First, we emphasize that the extension \emph{is not} achieved by setting it to identity on the orthogonal complement $ \bigoplus_{j=n+1}^\infty W_j $. Instead, the extension is characterized as follows: Suppose an operation $U$ acts nontrivially only on the first $n$ qubits so that, say, $U [G_{n,k}] = \sum_l z_{k,l} G_{n, l} $. Let $f\in L_2(0,1]$ be a general function. We write $f = \sum_{k} f_k$, where $f_k$ is the restriction of $f$ to the support of $G_{n,k}$. Then $U [f] = \sum_l z_{k,l} f_{l} $. In other words, operation $U$ does not alter the qualitative features of $f$ in the scales finer than $n$. As we will see, such operations are best described via the action of integral operators with distributional kernels. 
\vspace{.5cm}

\noindent
\emph{Examples.} Below, we will consider examples that involve the Pauli matrices  
\[
\sigma_x = \left(
  \begin{array}{cc}
    0 & 1 \\
    1 & 0 \\
  \end{array}
\right), \quad
\sigma_y = \left(
  \begin{array}{cc}
    0 & -i \\
    i & 0 \\
  \end{array}
\right), \quad
\sigma_z = \left(
  \begin{array}{cc}
    1 & 0 \\
    0 & -1 \\
  \end{array}
\right)
\]
as well as the matrices
\[
\sigma_+ = \left(
  \begin{array}{cc}
    0 & 0 \\
    1 & 0 \\
  \end{array}
\right),\quad
\sigma_- = \left(
  \begin{array}{cc}
    0 & 1 \\
    0 & 0 \\
  \end{array}
\right).
\]

Consider $\sigma_x$ as acting on the first qubit or on the second qubit, i.e., the maps
\[
\sigma_x^1 = \sigma_x \otimes I \otimes I\otimes I\ldots \mbox{ and } \sigma_x^2 = I\otimes \sigma_x \otimes I \otimes I\ldots.
\]
Both operations have nontrivial impact only on the first two qubits, e.g,
\[
\sigma_x^1 \left( z_{00}|00\rangle +z_{01}|01\rangle +z_{10}|10\rangle +z_{11}|11\rangle   \right) =
z_{00}|10\rangle +z_{01}|11\rangle +z_{10}|00\rangle +z_{11}|01\rangle,
\]
and similarly for $\sigma_x^2$.
The Borel isomorphism allows us to reinterpret this operation: 
\[
\sigma_x^1 \left( z_{00}G_{2,0} +z_{01} G_{2,1} +z_{10}G_{2,2} +z_{11} G_{2,3}   \right) =
z_{00}G_{2,2} +z_{01}G_{2,3} +z_{10}G_{2,0} +z_{11}G_{2,1}.
\]
Note that the map is fully defined by its restriction to $V_2$. (In the case of $\sigma_x^1$ we could have even selected $V_1$, which would result the same extension.) To see how these operators are extended from $V_2$ to the entire Hilbert space, it is convenient to view them as an integral operators with distributional kernels, as illustrated in Fig. \ref{Paulis2}.
Namely, the blue lines in the figure represent the support of the Dirac deltas in the square $(0,1] \times (0,1]$. Generally, these lines would be dressed by coefficients inherited from the operator. Since the Haar system consists of piecewise constant functions, the action of such operators is well defined on any basis function and, therefore, on $L_2(0,1]$. Fig. \ref{Paulis_action} illustrates how the map acts on the general state. As the operation it represents is a local operation on two qubits, it does not affect fine-scale features of the state function beyond the scale of $V_n$.
Of course, the Haar basis functions that encode those features are contained in $ \bigoplus_{j=2}^\infty W_j $. In this example, they are simply permuted in a certain way determined by the integral operator. In all cases, when a map is defined on $V_n$, it will be extended to the entire space in this manner. It is interesting to note that for an operation $\sigma_x^n$, the corresponding support of the measures consists of slant lines at dyadic scales of length $2^{-n}\sqrt{2}$. 

 It is also helpful to note the explicit formulas for gates. Indeed, recall that 
\[
f(x) = \int_{0}^{1} \delta(x-y)\ f(y) \, dy.
\]
Similarly,
\begin{eqnarray}
  \sigma_-^1[f](x)  &=& \chi_{(0,1]}(x)\, \int_{0}^{1}\delta(x-y - 1/2)\, f(y)\, dy\\
  \sigma_+^1[f](x)  &=& \chi_{(0,1]}(x)\, \int_{0}^{1}\delta(x-y + 1/2)\, f(y)\, dy,
\end{eqnarray}
etc. These are the ingredients that contribute to, say, $\sigma_x^{k}, \sigma_y^{k}$, via composition and scaling, see Fig. \ref{Paulis2}. 
\vspace{.5cm}

Another, alternative, characterization of the effect of such operators on general function is given in Lemma \ref{lemmarecursive}. 
In what follows $\epsilon_{k}(x) \in \{0,1\}$ denotes, the $k$-th digit in the dyadic expansion of $x\in(0,1]$, i.e.
$$
x = \sum_{k}\frac{\epsilon_{k}(x)}{2^k}.
$$
In order to avoid ambiguity, dyadic fractions always have an infinite expansion, e.g. $1/2 = 0.0111111\ldots$.
The following result is straightforward but crucially useful:

\begin{lemma}\label{lemmarecursive}
The Borel isomorphism \eqref{Borel_n}-\eqref{Borel_expl} translates the action of respective gates as follows:
\begin{equation}\label{sigma_plus_k}
  \sigma_{+}^{(k)}[f](x) = f\left(\left(x-\frac{1}{2^{k}}\right)\,\epsilon_{k}(x)\right),
\end{equation}
where $f$ is a continuous function in $(0,1]$ and, conventionally, we put $0$ for $f(0)$.
\begin{equation}\label{sigma_minus_k}
 \sigma_{-}^{(k)}[f](x) = f\left(\left(x+\frac{1}{2^{k}}\right)\, (1-\epsilon_{k}(x))\right).
\end{equation}
In addition, since
\begin{equation}
\sigma_{x}=\sigma_{+}+\sigma_{-},
\end{equation}
we also have
\begin{equation}\label{sigma_x_k}
  \sigma_{x}^{(k)}[f](x) = f\left(x+\frac{(-1)^{\epsilon_{k}(x)}}{2^{k}}\right).
\end{equation}
And, since
\begin{equation}
\sigma_{y}=i\sigma_{+}-i\sigma_{-},
\end{equation}
we have
\begin{equation}\label{sigma_y_k}
  \sigma_{y}^{(k)}[f](x) = i\,(-1)^{\epsilon_{k}(x)}\, f\left(x+\frac{(-1)^{\epsilon_{k}(x)}}{2^{k}}\right).
\end{equation}

\end{lemma}
\begin{proof}
  By direct inspection. 
\end{proof}

\vspace{1cm}

\noindent
\emph{Example 2.} The quantum Fourier transform on $n$ qubits is defined as the map
\[
\mathcal{F}_n:\,\,|l\rangle \mapsto 2^{-n/2}\, \sum\limits_{k =0}^{2^n -1} e^{2\pi i kl/2^n } \,|k\rangle \quad \mbox{ for all } l  \in\{ 1,2,\ldots 2^n-1\}.
\]
where we rely on the shorthand notation (\ref{short}). This can be reinterpreted via the  Borel isomorphism (\ref{Borel_short}) as the map $\mathcal{F}_n^B = B\,\mathcal{F}\, B^{-1}$. Namely,
\begin{equation}\label{QFT_Borel}
\mathcal{F}_n^B:\,\,  G_{n,l} \mapsto 2^{-n/2}\, \sum\limits_{k =0}^{2^n -1} e^{2\pi i kl/2^n } \, G_{n,k}\quad \mbox{ for all } l  \in\{ 1,2,\ldots 2^n-1\}.
\end{equation}
Naturally,  $\mathcal{F}_{n}^B$ acts non-trivially in $V_n$ and extends, as discussed to the whole space. $\mathcal{F}_{n}^B$ is a unitary map in $L_2(0,1]$.
Suppose, for a fixed $n>1$, $f_d = \sum_{k} f_k\, G_{n,k}$. (This is a typical situation when a function is discretized for the sake of numerical analysis. For example, if $f(x)$ is a continuous function in $(0,1]$, we may set $f_k = f(x_k)$ where $x_k\in I_{n,k}$ is some choice of a point for each of the $2^n$ dyadic intervals.) A direct examination shows that $\mathcal{F}_n^B f_d $ is exactly the discrete Fourier transform of the discrete function $f_d$.

It is known that, in principle, the execution of  $\mathcal{F}_n$ may be carried out via $O(n^2)$ local quantum gates. By comparison, the fastest classical algorithm for the discrete Fourier transform, known as the Cooley-Tukey algorithm to execute  $\mathcal{F}_{n}^B$ numerically still requires $O(n2^n)$ arithmetical operations. Similar features can be displayed for a number of other algorithms, where the complexity of the quantum gate circuit implementing it spectacularly undercuts the required number of arithmetical operations of the classical implementation. This promises a leap in computational efficiency when and if the transition to quantum information processing becomes a reality. Of course, this hinges on advances in the technologies for manipulating the quantum states in physical systems, which are being strived for at many current era centers of research.

\begin{figure}[!ht]
\centering
  \includegraphics[width=120mm]{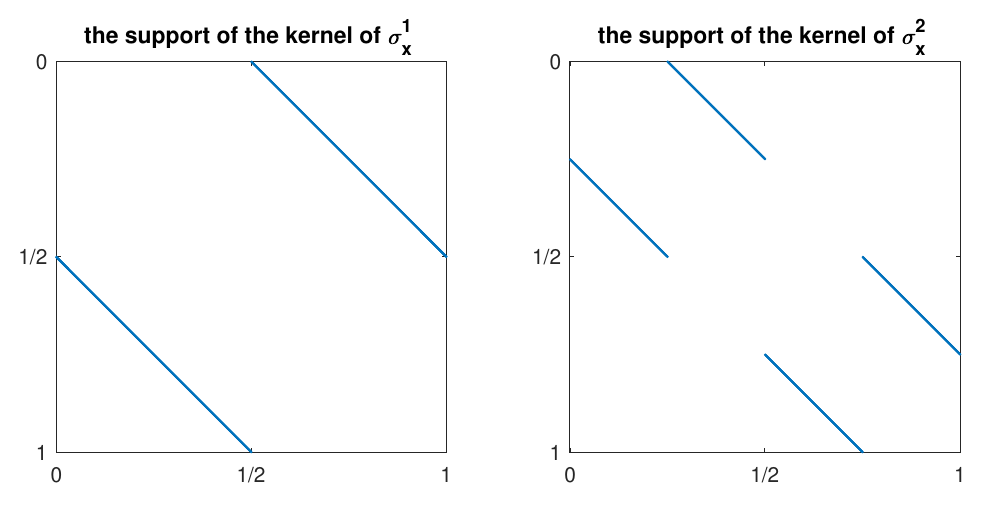}
\caption{\small Pauli gates $\sigma_x$ can be represented as integral operators acting in $L_2(0,1]$. The kernels of these operators are Dirac measures supported in the blue lines.  Note the inverted direction of the y-axis.}
\label{Paulis2}
\end{figure}

\begin{figure}[ht!]
\centering
\includegraphics[width=120mm]{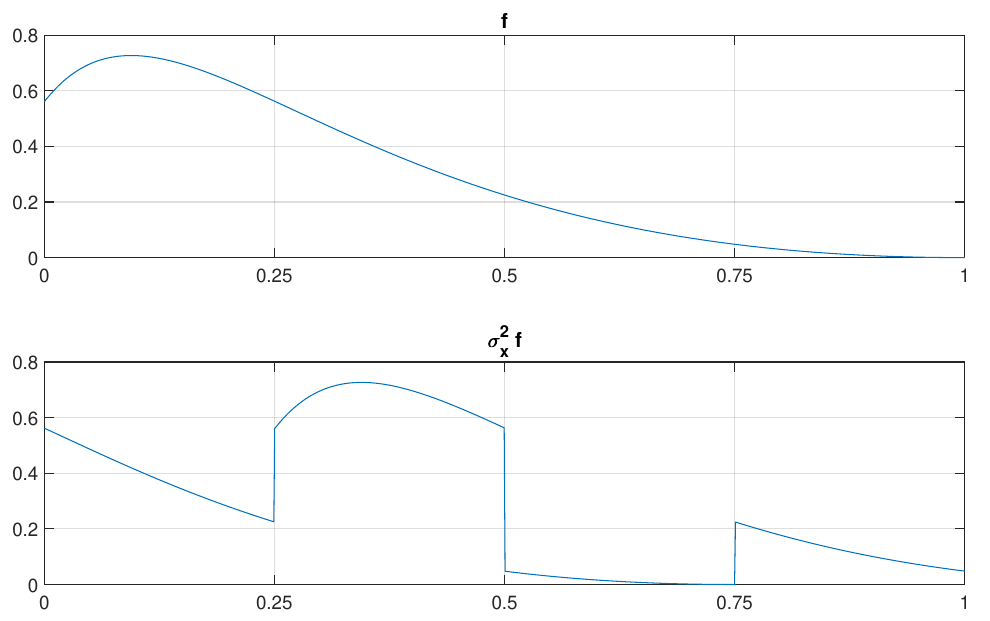}
\caption{The action of $\sigma_x^2$ on a general state function $f\in L_2(0,1]$. This illustrates the general principle behind extending a quantum operation from $V_n$ to the full Hilbert space. }
\label{Paulis_action}
\end{figure}

\section{Some aspects of non-local quantum operations}
In this section, we aim to present a comprehensive introduction to the non-local quantum operation, focusing on the operator denoted as $C_y$. Our approach culminates in  a theorem the key properties of $C_y$ and other, related operators. The Borel isomorphism and the Haar transform are foundational for our approach. Using these tools $C_y$ can be reduced to a sum of an operator that admits a multi-scale block structure and a compact ``remainder". We refer to Section \ref{section_appendix} for the heuristic arguments which shed light on the essential nature of the main problem.
\subsection{Setting the stage}\label{subsection_stage}

We begin with the operator\footnote{It often is convenient to use the term periodization, i.e., one can refer to $C_x$ as periodization of $\sigma_x$. Substituting different operators of $\sigma_x$ one obtains different periodizations.}
\begin{equation}\label{C_x}
 C_x = \sum_{k =1}^{\infty} \frac{1}{2^k} \, \sigma_x^k,
\end{equation}
originally explored in \cite{R1}.
The Borel isomorphism allows us to represent $C_x$ as an integral operator with a distributional kernel. The fractal nature of the kernel, as illustrated in Fig. \ref{Fig_C_x}, reflects the fact that $C_x$ acts nontrivially across all scales, i.e., acts nontrivially on all qubits. It follows from Lemma \ref{lemmarecursive} that the action of the operator on, say, continuous function $f$  is given by an  explicit formula:
 \begin{equation}\label{def_C}
  C_x[f](x) = \sum_k \frac{1}{2^k}f\left(x  + \frac{(-1)^{\epsilon_k(x)}}{2^k}\right)\quad x \in (0,1].
\end{equation}

\begin{figure}[ht!]
\begin{center}
\includegraphics[width=60mm]{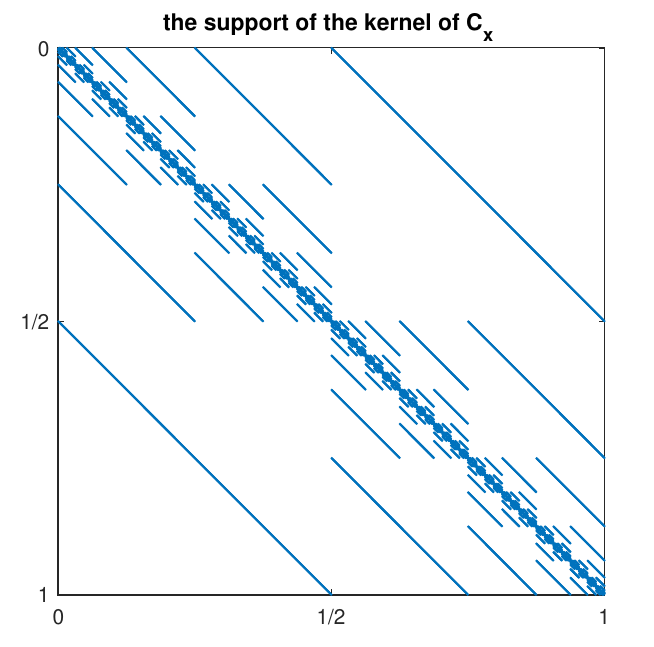}
\caption{The blue lines mark the supports of the Dirac measures in the kernel of the operator $C_x$ as defined in (\ref{C_x}). The support of the operator $C_y$, defined in (\ref{C_y}), is the same.}
\label{Fig_C_x}
\end{center}
\end{figure}

$C_x$ was characterized in \cite{R1}. We summarize those findings: 
Let $D_n$ be the matrix of $\Pi_n C_x \Pi_n$ in the canonical basis $(H_{n,k})_{k=0}^{2^n-1}$. Then\footnote{We generally denote the $k\times k$ identity matrix by $I_k$, but skip the index when the matrix size is clear from context.},
  \begin{equation}\label{block_D}
  D_0 = [0], \quad \mbox{ while } \quad
    D_{n+1} = \frac{1}{2} \left[
                            \begin{array}{cc}
                              D_n & I \\
                              I & D_n \\
                            \end{array}
                          \right]\quad \mbox{ for } n\geq 0.
    \end{equation}
Also, if $n>0$, then $D_n$ is invertible, and its complete list of the eigenvalues is
\begin{equation}\label{eigsD}
\{\pm (2k+1)/2^n: \, k = 0, 1, \ldots 2^{n-1}-1\}.
\end{equation}
$C_x$ preserves the direct sum decomposition (\ref{direct_sum}), i.e.  $C_x V_0 = V_0$ and $C_x W_n \subseteq W_n$ for all $n$. In fact,
\begin{equation}\label{blocks_C}
C_x  =\mathcal{T_H}^\dagger \left(I_1\oplus \bigoplus_{n=0}^\infty \, D_n \right)\mathcal{T_H}.
\end{equation}

\vspace{.2cm}

One of the partial goals of this work is to understand how to obtain results as the above for more general operators. It is not a trivial matter. To explain the nature of the challenge, let us address the case of
\begin{equation}\label{C_y}
 C_y = \sum_{k =1}^{\infty} \frac{1}{2^k} \, \sigma_y^k.
\end{equation}
Again, taking advantage of the Borel isomorphism $C_y$ can be represented via an integral operator with the distributional kernel.  The location of non-trivial measures is precisely as that of $C_x$ shown in Fig. \ref{Fig_C_x}. However, this time the measures are weighted by complex coefficient. A formula analogous to (\ref{def_C}) can be deduced from Lemma \ref{lemmarecursive}, namely:
 \begin{equation}\label{def_C_y-phi}
  C_y[f](x) = \sum_k \frac{(-1)^{\epsilon_k(x)} i}{2^k} f \left(x  + \frac{(-1)^{\epsilon_k(x)}}{2^k}\right)\quad x \in (0,1].
\end{equation}
It is convenient to consider additional operators:
\begin{equation}\label{C_pm}
 C_- = \sum_{k =1}^{\infty} \frac{1}{2^k} \, \sigma_-^k, \quad 
  C_+ = \sum_{k =1}^{\infty} \frac{1}{2^k} \, \sigma_+^k,
\end{equation}
see Fig. \ref{periodized}. Since
$\sigma_y = i(\sigma_+ - \sigma_{-}),$
it follows that
\begin{equation}\label{eqPy}
C_y = i(C_{+} - C_{-}).
\end{equation}
The crucial ingredient in the analysis of operators $C_-, C_+$, and $C_y$ is the operator defined by the kernel\footnote{See Section \ref{section_appendix} for an outline of the heuristics and numerical experiments that lead one to the discovery of the role of operator $L$.} 
\begin{equation}\label{ell_kernel}
\ell(x,y) =
\left\{
    \begin{array}{lr}
        1, & \text{if } x < y        \\
        0, & \text{if } x\geq y   \\
    \end{array}\right..
\end{equation}
Thus, the operator itself is
\begin{equation}\label{ExplicitL}
L[u](x) = \int_{0}^{1}\ell(x,y)u(y)\,dy = \int_{0}^{x}u(y)\,dy,
\end{equation}
which is precisely the anti-derivative. It is now natural to conjecture that when represented in the Haar basis, both operators $C_{-}+L$ and $C_{+}+L'$ turn into infinite block matrices.


\begin{figure}[ht!]
\begin{subfigure}{0.5\textwidth}
\includegraphics[width=1\linewidth, height=7cm]{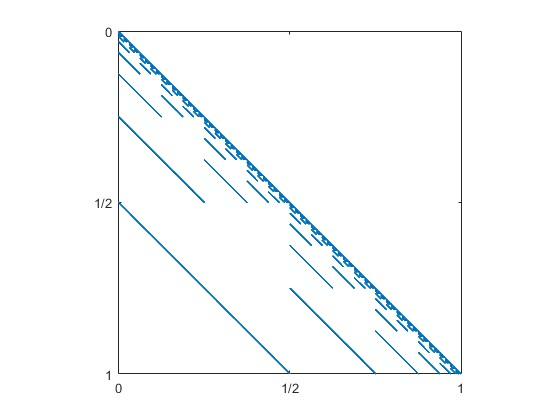} 
\caption{The support of the kernel of $C_{+}$}
\label{fig:C+}
\end{subfigure}
\begin{subfigure}{0.5\textwidth}
\includegraphics[width=1\linewidth, height=7cm]{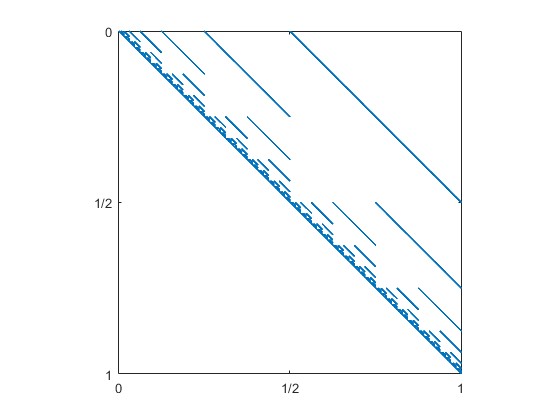}
\caption{The support of the kernel of $C_{-}$}
\label{fig:C-}
\end{subfigure}
\caption{The kernels of $C_+$ and $C_-$ are scaled Dirac measures whose supports are shown by the blue lines.}
\label{periodized}
\end{figure}

\subsection{Launching rigorous analysis}

The conjecture of previous subsection readily translates into a conjecture about the structure of $C_y$. Indeed, \eqref{eqPy} gives
\begin{equation}\label{EqPy}
C_y+i(L'-L)=i(C_{+} -C_{-}+L'-L)\\
           =i((C_{+}+L')-(C_{-}+L)).
\end{equation}
Now we can simplify further and for convenience, we introduce the operators
\begin{equation*}
K=i(L'-L),
\end{equation*}
\begin{equation*}
P_{-}=(C_{-}+L),
\end{equation*}
\begin{equation}
P_{+}=(C_{+}+L'),
\end{equation}
Note that both $L$ and $L'$ having square-integrable kernels are Hilbert-Schmidt operators. Automatically, $K$ is also compact. This gives the main object of our study the following structure:
\begin{equation}\label{Cy_structure}
  C_y = i\, (\, P_+ - P_-) - K.
\end{equation}
As we will see, the first term yields itself to a block representation via the Haar transform. To this end we will examine
\begin{equation}\label{FormulaCy}
 \mathcal{T_H}[i(P_{+}-P_{-})]\mathcal{T_H}'.
\end{equation}
In analogy with $\eqref{def_C_y-phi}$, and according to Lemma \ref{lemmarecursive}, operators $C_{-}$ and $C_{+}$ are defined first giving its explicit construction for a continuous or piecewise constant test function $f (x)$; namely:
 \begin{equation}\label{def_C_+-phi}
  C_{-}[f](x) = \sum_k \frac{1}{2^k}f\left(\left(x  + \frac{1}{2^k}\right)\,(1-\epsilon_k(x))\right)\quad x \in (0,1],
\end{equation}
 \begin{equation}\label{def_C_--phi}
  C_{+}[f](x) = \sum_k \frac{1}{2^k}f\left(\left(x  - \frac{1}{2^k}\right)\,\epsilon_k(x)\right)\quad x \in (0,1].
\end{equation}
Recall that the characteristic function $\chi_{(0,1]} $ is the first element of the Haar basis. We have the following facts:
\begin{lemma}\label{lemmaChi}
$\chi_{(0,1]} $ is the simultaneous eigenvectors of both $P_-$ and $P_+$. Moreover,
\begin{enumerate}

\item
$P_{-}[\chi_{(0,1]}] = \chi_{(0,1]}$,

\item
$P_{+}[\chi_{(0,1]}] = \chi_{(0,1]}$.

 \end{enumerate}
\end{lemma}
\begin{proof}
\begin{enumerate}
\item
First, we examine $C_{-}[\chi_{(0,1]}](x)$:
   \[
   C_{-}[\chi_{(0,1]}](x) = \sum_k \frac{1}{2^k}\chi_{(0,1]}\left((x + \frac{1}{2^k})(1-\epsilon_k(x))\right).
   \]
Note that the $k$-th term of the sum contributes only if the corresponding $\epsilon_k(x) = 0$. In that case $x+1/2^k \in (0,1]$, which implies $\chi_{(0,1]}\left((x + \frac{1}{2^k})(1-\epsilon_k(x))\right) = 1$. Thus,
\[
 C_{-}[\chi_{(0,1]}](x) = \sum_{k: \epsilon_k(x) = 0} \frac{1}{2^k} = 1-x.
\]
At the same time
   \[
   L[\chi_{(0,1]}](x) = \int_0^x \chi_{(0,1]}(y)dy = x.
   \]
Combining the results above we obtain
$P_{-}[\chi_{(0,1]}](x) = C_{-}[\chi_{(0,1]}](x) + L[\chi_{(0,1]}](x) = 1- x + x \equiv 1
 $ for all $x\in (0,1]$.

 \item
Similarly as above, we calculate both terms of $P_{+}[\chi_{(0,1]}]$.  First,
   \[
   C_{+}[\chi_{(0,1]}](x) = \sum_k \frac{1}{2^k}\chi_{(0,1]}\left((x - \frac{1}{2^k})\, \epsilon_k(x)\right).
   \]
   Note that the $k$-th terms contributes to the sum nontrivially only if $\epsilon_k(x) = 1$, in which case $x-1/2^k \in (0,1]$. This in turn implies
   \[
   C_{+}[\chi_{(0,1]}](x) = \sum_{k: \epsilon_k(x) =1} \frac{1}{2^k} = x.
   \]
   Next, we examine $L'[\chi_{(0,1]}](x)$, where $L'$ is the transpose of operator $L$. The transpose of an integral operator with kernel $\ell(x,y)$ is defined as:
   \begin{equation}
   L'[u](x) = \int_0^1 \ell(y,x)u(y)dy.
   \end{equation}
 Considering the kernel $\ell(x,y) = \begin{cases} 1, & \text{if } x < y \\ 0, & \text{if } x \geq y \end{cases}$, we can see that $\ell(y,x)$ is equal to 1 only when $y > x$. Thus, we have:
   \begin{equation}
   L'[\chi_{(0,1]}](x) = \int_x^1 \chi_{(0,1]}(y)dy = \int_x^1 dy = 1 - x.
   \end{equation}
Combining these results, we obtain
 $  P_{+}[\chi_{(0,1]}](x) = C_{+}[\chi_{(0,1]}](x) + L'[\chi_{(0,1]}](x) = x + 1 - x = 1
 $ for all $x\in (0,1]$.
This completes the proof.
\end{enumerate}

\end{proof}

\subsection{Statement of the Theorem}
 We start by stating a theorem characterizing the properties of $C_y$ via the elements of its decomposition (\ref{Cy_structure}).
\begin{theorem}\label{P+theor}
Operators $P_{-}$ and $P_{+}$ have the following properties:
\begin{enumerate}
\item
 Let $D_{-}^{(n)}$ denote the matrix of $\Pi_n P_{-} \Pi_n$ in the canonical basis $(H_{n,k})_{k=0}^{2^n-1}$. Then,
\begin{equation}\label{blocks_Cy}
P_{-}  =\mathcal{\mathcal{T_H}}^\dagger \left([1] \oplus \bigoplus_{n=0}^\infty \, D_{-}^{(n)} \right)\mathcal{\mathcal{T_H}},
\end{equation}
i.e., $\mathcal{T_H} P_{-}\mathcal{T_H}' = \mathcal{T_H}(C_{-}+L)\mathcal{T_H}'$ has block structure. In particular, $P_{-}$ preserves the direct sum decomposition (\ref{direct_sum}), i.e.  $P_{-}V_0 = V_0$ and $P_{-} W_n \subseteq W_n$ for all $n$.

\item
 The following recurrence determines the sequence of blocks:
\begin{equation}\label{blockD-n}
D_{-}^{(0)} = \left[ \begin{array}{cc}
0
\end{array}\right], \quad \text{and} \quad  D_{-}^{(n+1)} =  \frac{1}{2}\left[
                            \begin{array}{cc}
                              D_{-}^{(n)} & I \\
                              0 & D_{-}^{(n)} \\
                            \end{array}
                          \right] \quad \quad \text{for}\quad n\geq 0
\end{equation}
where $0$ and $I$ stand respectively for the trivial and the identity matrices of size $n\times n$.

\item
Analogous statements hold for $P_{+}$.

\item
$P_{-}$ and $P_{+}$ can be extended to continuous operators $P_{+}, P_{-}: L_2(0,1] \rightarrow L_2(0,1]$, which are mutually conjugate.

\end{enumerate}
\end{theorem}

\subsection{Proof of Theorem \ref{P+theor}}

The proof of Theorem 3.1 we be based on Lemma \ref{lemmarecursive}, and also the following fact:
\begin{lemma}
Let $\mathbb{H}\ominus V_0 = \bigoplus\limits_{K=1}^{\infty} W_k $. The following identity characterizes the adjoint of $L$:
\begin{equation}
L'|_{\mathbb{H}\ominus V_0}= -L|_{\mathbb{H}\ominus V_0}.
\end{equation}
\end{lemma}

\begin{proof}
Let $f, g \in \mathbb{H}\ominus V_0$, so that $\int_{0}^{1} f(x) dx = 0$ and $\int_{0}^{1} g(x) dx = 0$. Integrating by parts we obtain
\begin{equation*}
\langle g | L f \rangle  = \int_{0}^{1}g^* (x) \, \int_{0}^{x} f (t) dt \, dx
=  - \int_{0}^{1} \int_{0}^{x} g^*(t)dt f(x) dx = \langle -Lg | f \rangle,
\end{equation*}
as claimed.
\end{proof}

We are now in a position to complete the proof of the theorem:

\begin{proof}\label{proof}
In light of Lemma \ref{lemmarecursive} the first block of the matrix of $P_-$ in the Haar basis is, indeed, the $1$-by-$1$ matrix $[1]$. Next, we determine $P_- [\, H_{0,0} ]$ where $H_{0,0} =H = \chi_{(0, 1/2]} - \chi_{(1/2, 1]}$ as in (\ref{the_Haar_fn}). Recall that $P_- = C_- + L$. We have seen in the proof of Lemma \ref{lemmarecursive} that $C_-[\chi_{(0,1]}] = 1-x$. To evaluate $C_- [\chi_{(0, 1/2]}]$ we take advantage of self-similarity, see Fig. \ref{periodized}. We see directly from the graph that
\[
C_-[\chi_{(0,1/2]}] = \frac{1}{2}\, (1-2x)\, \chi_{(0,1/2]}.
\]
Note that the right hand side represents a scaled copy of $1-x$ placed in the first half of the unit interval. Similarly,
\[
C_-[\chi_{(1/2,1]}] = \frac{1}{2}\,\chi_{(0,1/2]} +  (1-x)\,\chi_{(1/2,1]},
\]
i.e., the result is a copy of $\chi_{(0,1/2]} $ itself shifted to the left half, followed by a scaled copy of $1-x$ placed in the right half of the unit interval. Thus,
\[
C_- [H] = C_-[\chi_{(0,1/2]}] - C_-[\chi_{(1/2,1]}] = - \frac{1}{2} |1-2x| = - L [H] ,
\]
so that $P_- [H] = (C_- +L) [H] = 0$. This proves that $P_- W_0 \subset W_0$, and the block $D_-^{(0)} = [0]$, i.e. it  is indeed trivial as in (\ref{blockD-n}).

As the next step, consider  $W_1$, which is spanned by $H_{1,0}$ and $H_{1,1}$. In order to consider the action of $P_-$ on $H_{1,0}$ we only need to consider the upper-left quarter of the kernel of $C_-$. Due to self-similarity, it is clear that $P_- [H_{1,0}] = 0$ as, indeed, this directly caries over from the last step. Next, consider $P_- [H_{1,1}]$. The action of the lower-right quarter of the kernel of $C_-$ (in addition to the action of $L$) again, gives $0$. However, an application of the diagonal of the lower-left quarter results in a scaled replica of $H_{1,0}$ in the first half of the unit interval, i.e., $P_- [H_{1,1}] = H_{1,0} /2$. This demonstrates that $P_- W_1 \subset W_1$, with the matrix of $P_-$ restricted to this space being precisely
\[
\frac{1}{2}\left[
                            \begin{array}{cc}
                             0 & 1 \\
                              0 & 0 \\
                            \end{array}
                          \right] \quad \quad \text{for}\quad n\geq 0.
\]
In light of self-similarity and the dyadic structure of Haar scales, this pattern repeats at all scales, leading to recursion (\ref{blockD-n}).

\end{proof}

%



\subsection{Extensions to the Bloch sphere} \label{secBS}

Let, $(\alpha,\beta,\gamma)$ be a point on the sphere, i.e., $\alpha^2 + \beta^2 + \gamma^2 =1$. Recall that a point on the Bloch sphere is the matrix:
$$
\sigma_{(x, y, z)}=\alpha\, \sigma_x +\beta\, \sigma_y+\gamma\,\sigma_z.
$$
A point on the Bloch sphere can also be associated with the periodized operator $C_{(x, y, z)}=\alpha\ C_x +\beta \, C_y +\gamma \, C_z$. It is an interesting object for further studies. However, the results of previous section are insufficient to address it in full generality. Indeed, $C_z$ has been identified in \cite{R1} as a multiplier; namely, 
\[
C_z[f](x) = (1-2x)\, f(x)\,\, \mbox{ for all } \,\, f\in L_2(0,1]. 
\]
Therefore, $C_z$ does not have eigenvalues. Its spectrum consists of the range of the multiplier function which, again, is the interval $[-1,1]$. Because of that analysis of the spectrum of $C_{(x, y, z)}$ is a new challenge, and the solution will likely require new insights. However, it is interesting to focus on the equator of the Bloch sphere. As we have seen operators $C_x$ and $C_y +K$ reduce to block matrices and so does their linear combination. This warrants a few comments, namely consider the operator:
\begin{equation}\label{Equator}
  C_{\theta}+ \sin(\theta) \,K = \cos(\theta) \, C_x+ \sin(\theta) \,[C_y+K].
\end{equation}
where $\theta\in [0, 2\pi)$ parameterizes the equator in the Bloch sphere.
 Based on the findings of the previous sections, the right-hand side is a block matrix, specifically: 
\[
\cos(\theta )\left([1] \oplus \bigoplus_{n=0}^\infty \, D_n \right) + i\sin(\theta)
 \left([0] \oplus \bigoplus_{n=0}^\infty \, (D_{+}^{(n)} - D_{-}^{(n)})\right) .
\]
It is convenient to rewrite this in a more compact form:
\[
\mathcal{T_H} \,[ C_{\theta} + \sin(\theta) \,  K] \, \mathcal{T_H}'= [\cos\theta] \oplus \bigoplus_{n=0}^\infty \, D_n^\theta .
\]
Collecting the findings of Theorem \ref{P+theor}, Corollary \ref{cor_theta}, and formulas \eqref{block_D}, \eqref{blocks_C},  we obtain the following recurrence relation for the blocks:
  \begin{equation}\label{block_D_theta}
  D_0^\theta = [0], \quad \mbox{ while } \quad
    D_{n+1}^\theta = \frac{1}{2} \left[
                            \begin{array}{cc}
                              D_n^\theta & e^{i\theta} I \\
                              e^{-i\theta} I & D_n^\theta \\
                            \end{array}
                          \right]\quad \mbox{ for } n\geq 0.
    \end{equation}   

\begin{lemma} The set of eigenvalues of $D_n^\theta$ is:
  \[
    \{ \pm(2k+1)/2^n : k=0,1,... 2^{n-1}-1 \}.
    \]
    Let the corresponding eigenvectors be denoted $v_{n,k, \pm}$. The eigenvectors satisfy the recurrence relation 
      \[
v_{1, 0, \pm} = \left[
                            \begin{array}{c}
                              e^{i\theta}  \\
                              \pm 1 \\
                            \end{array}
                          \right], \quad 
 v_{n+1, \pm k \pm 2, \pm} = \left[
                            \begin{array}{c}
                              e^{i\theta} v_{n,k, \pm} \\
                              \pm  v_{n,k, \pm}\\
                            \end{array}
                          \right].                         
    \] 
    Note that the number of eigenvalues and eigenvectors is doubled when passing from $n$ to $n+1$. 
\end{lemma}
\begin{proof}
  Observe that the matrix
 \[
D_{1}^\theta = \frac{1}{2} \left[
                            \begin{array}{cc}
                              0 & e^{i\theta}  \\
                              e^{-i\theta}  & 0 \\
                            \end{array}
                          \right]
\]
has eigenvalues $\pm 1/2$ with the corresponding eigenvectors $v_{1,0,\pm}$ as claimed. 
Moreover, the recurrence (\ref{block_D_theta}) indicates that 
\[
D_{n+1}^\theta \,  \left[
                            \begin{array}{c}
                              e^{i\theta} v_{n,k, \pm} \\
                              \pm  v_{n,k, \pm}\\
                            \end{array}
                          \right]= \frac{1}{2}\left(\frac{\pm(2k+1)}{2^n} \pm 1 \right) 
                          \left[
                            \begin{array}{c}
                              e^{i\theta} v_{n, \pm} \\
                              \pm  v_{n, \pm}\\
                            \end{array}
                          \right].
\]
This furnishes the induction argument that proves the lemma. 
\end{proof}

We obtain the following 
\begin{corollary} \label{cor_theta}
  We have the following facts:
  \begin{enumerate}
    \item The complete list of eigenvalues of $ C_{\theta} + \sin(\theta) \,  K$ is:
    \[
    \{ \pm(2k+1)/2^n : k=0,1,... 2^{n-1}-1 \}.
    \]
    \item The spectrum of $ C_{\theta} + \sin(\theta) \,  K$ is the closure of the set of eigenvalues, i.e., 
    \[
    \Sigma\, ( C_{\theta} + \sin(\theta) \,  K) = [-1, 1].
    \]
    \item The essential spectrum \footnote{The essential spectrum of an operator remains unchanged when the operator is perturbed by a compact operator.} of $C_\theta$ is the interval $[-1,1]$.
     \end{enumerate}
\end{corollary}
\vspace{.5cm}

\noindent 
\emph{Remark.} Note that the compact correction disappears at exactly two point on the Bloch sphere's equator, that is at $\theta =0, \pi$. Note that $C_{\theta = 0 } = C_x$, and $C_{\theta = \pi } = -C_x$. Of course, $\|C_x\|=1$. It is also interesting to observe the following: 
\begin{proposition} \label{prop_norm}
$\|C_\theta\| \leq 1$.
\end{proposition} 
\begin{proof}
In light of (\ref{def_C}) and (\ref{def_C_y-phi}), for an $f \in L_2(0,1]$, we have:
\begin{equation*}
\|C_\theta[f]\|^2 = \int_0^1 \left|\sum_k \frac{\cos \theta + (-1)^{\epsilon_k(x)}i\sin\theta  }{2^k}f\left(x + \frac{(-1)^{\epsilon_k(x)}}{2^k}\right)\right|^2 dx \\
\end{equation*}
\begin{equation*}
\leq \int_0^1 \left(\sum_k \frac{1 }{2^k}\left|f\left(x + \frac{(-1)^{\epsilon_k(x)}}{2^k}\right)\right|\right)^2 dx \\
\end{equation*}
and so, because of the convexity of the parabola,
\begin{equation*}
\leq \int_0^1 \sum_k \frac{1}{2^k}\left|f\left(x + \frac{(-1)^{\epsilon_k(x)}}{2^k}\right)\right|^2 dx \leq \|f\|^2.
\end{equation*}
Therefore, $||C_\theta\| \leq 1$.
\end{proof}


\begin{figure}[ht!]
\begin{center}
\includegraphics[width=150mm]{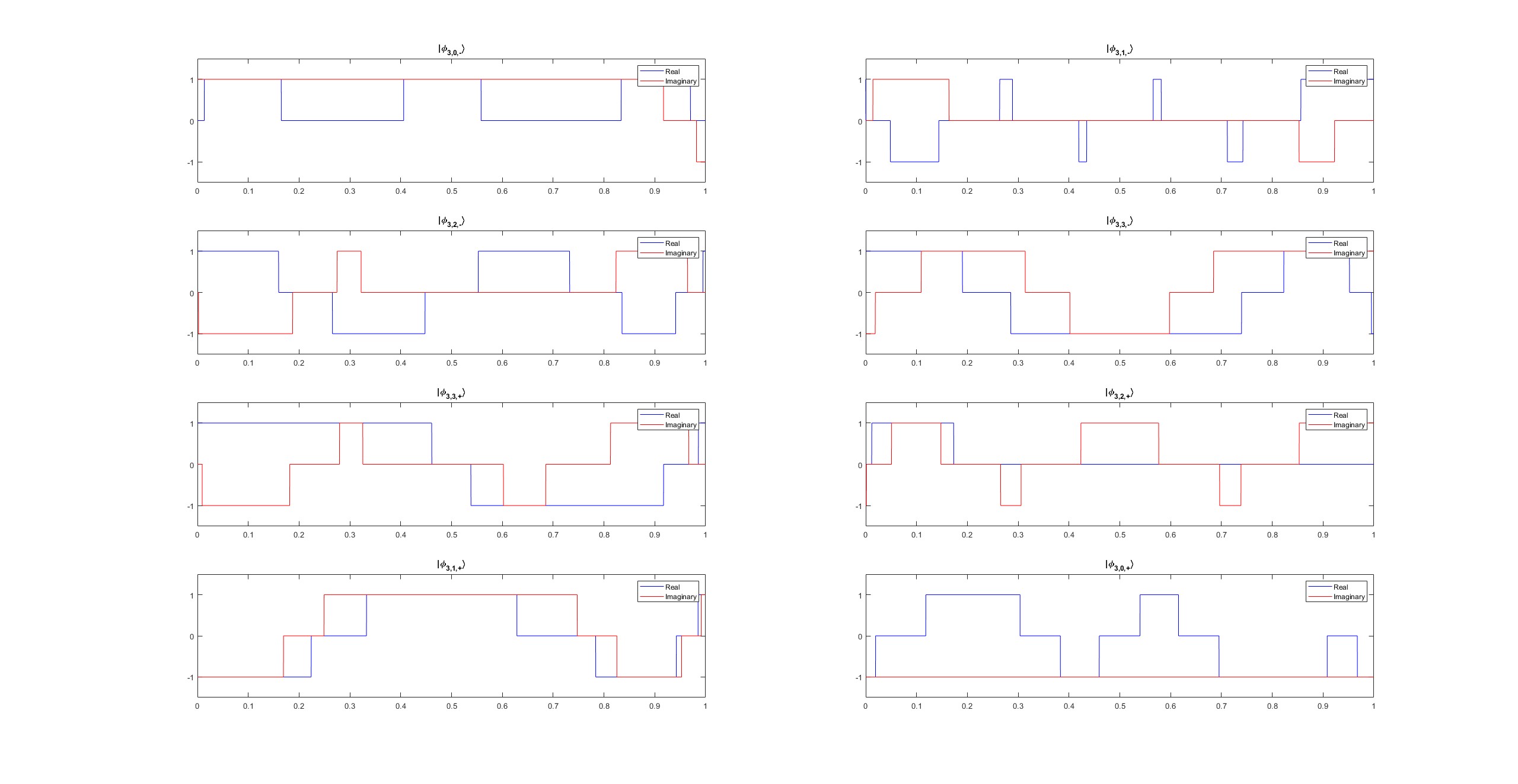}
\caption{Representation of the real and imaginary parts of the eigenstates of operator $C_y+K$ corresponding to the eigenvalues $\pm \frac{7}{8}, \pm \frac{5}{8}, \pm \frac{3}{8}, \pm \frac{1}{8}$.
}
\label{ESCY}
\end{center}
\end{figure}


\section{The dynamics of an infinite qubit array}

We now turn attention to a generalized Jaynes-Cummings model. Recall the Hamiltonian governing the interaction between a linear oscillator, representing a single mode of electromagnetic field, and a two-level system (qubit), see e.g., \cite{Shore_Knight, Gerry-Knight}:
\begin{eqnarray}
\mathcal{H} = I\otimes \mathcal{H}_F + \mathcal{H}_q\otimes I +\mathcal{H}_I:\quad
\mathbb{H}_Q\otimes\mathbb{H}_F \longrightarrow
\mathbb{H}_Q\otimes\mathbb{H}_F,
\end{eqnarray}
in which $\mathbb{H}_F= L_2(\mathbb{R})$ is the oscillator's Hilbert space. Furthermore,
\begin{eqnarray}
  &&\mathcal{H}_F = \hbar\omega\, (\hat{a}^\dagger \hat{a} + 1/2)   \quad  \\
 &&\mathcal{H}_q =  \frac{\hbar\Omega}{2}(|e\rangle\langle e| -|g\rangle\langle
g| ) = -\frac{\hbar\Omega}{2}\sigma_z,   \quad 
\end{eqnarray}  
whereas the qubit-field interaction may assume the form
\[
\mathcal{H}_I = \lambda\hbar\,\left(\cos(\theta) \, \sigma_x+ \sin(\theta) \, \sigma_y\right)\otimes (\hat{a} + \hat{a}^\dagger).
\]
(Note that the case of $\theta = 0$ is the one originally considered in \cite{R1}.)
Naively, the generalized interaction term for an infinite array could be postulated in the form
\[
\lambda\hbar \, \left(\cos(\theta) \, C_x+ \sin(\theta) \, C_y\right)\otimes (\hat{a} + \hat{a}^\dagger) =
\lambda\hbar \, C_\theta\otimes (\hat{a} + \hat{a}^\dagger)
\]
However, the form more amenable to analysis is: 
\[
\mathcal{H}_I = \lambda\hbar\, (C_{\theta}+ \sin(\theta) \,K) \otimes (\hat{a} + \hat{a}^\dagger)
\]
Similarly, the Hamiltonian $\mathcal{H}_q$, when generalized to an infinite array, assumes the form 
\begin{eqnarray}
V = \sum_{k=1}^{\infty} -\frac{\hbar\Omega_k}{2}\sigma_z^{k}, \quad
\text{ where } \Omega_k = \frac{\Omega}{2^k}.
\end{eqnarray}
It was established in \cite{R1} that in fact
\begin{eqnarray}
V[f](x) = V(x)f(x)=\left(x-\frac{1}{2}\right)f(x).
\end{eqnarray}
Thus, the total Hamiltonian has the form
\begin{eqnarray}
\mathcal{H}_{QMM} = I\otimes \mathcal{H}_F  +\lambda\hbar (C_{\theta}+ \sin(\theta) \,K) \otimes(\hat{a} + \hat{a}^\dagger) + V(x)\otimes I.
\end{eqnarray}
Furthermore, we pass to the regime where  $\Omega$ (representing the scale of qubit excitation energies) is negligible compared to the energy of the field mode, $\omega$, and the qubit-field interaction scale, $\lambda$, i.e., we focus on
\begin{eqnarray}
\mathcal{H}'_{QMM} = I\otimes \mathcal{H}_F +\lambda\hbar (C_{\theta}+ \sin(\theta) \,K) \otimes(\hat{a} + \hat{a}^\dagger).\label{H'}
\end{eqnarray}
With these assumptions in place, analysis of the dynamics may be carried out in the same manner as in the case $\theta = 0$ described in \cite{R1}. For the reader's convenience we summarise those arguments here.  The mixed state of the system satisfies the Heisenberg equation, i.e.,
\begin{equation}\label{Heis_full}
  i\partial_t \hat{\rho} =[\mathcal{H}'_{QMM}, \hat{\rho}].
\end{equation}
We search for solutions in the form
\[
\hat{\rho}=|\Phi_{n,k,s,\theta}\rangle \langle \Phi_{n,k,s, \theta}|\otimes \hat{\rho}_F .
\]
Here $\hat{\rho}_F:\mathbb{H}_F \longrightarrow
\mathbb{H}_F$, $s = \pm $ and $\Phi_{n,k,s,\theta}$ is an eigenstate of $C_{\theta}+ \sin(\theta) \,K$ corresponding to the eigenvalue
$E_{n,k,s,\theta}=s(2k+1)/2^n$, where $k = 0, 1, 2, ...2^{n-1}-1$ for $n = 1, 2,...$. Examples of such eigenfunctions, when $\theta = \frac{\pi}{3}$, are given in Fig.\ref{Cthetapi3}. Equation \eqref{Heis_full} implies that $\hat{\rho}_F$ evolves via 
\begin{equation}\label{Heis_F}
  i\partial_t \hat{\rho}_F =[\mathcal{H}_F +\lambda E_{n,k,s,\theta} (\hat{a}+\hat{a}^\dagger), \hat{\rho}_F].
\end{equation}
In the Wigner representation, see e.g., \cite{Gardiner-Zoller}, \cite{Gosson}, the operator $\hat{\rho}_F$  is replaced by a function of two real variables given by
\[
f(q, p) = \int_{\mathbb{R}}d\xi_1 \int_{\mathbb{R}}d\xi_2 \, e^{-2\pi i(\xi_1 q+\xi_2 p)}\, Tr(W_{\xi_1, \xi_2}\hat{\rho}_F),
\]
where $W_{\xi_1, \xi_2}= \exp (2\pi i(\xi_1 \hat{q}+\xi_2 \hat{p}))$. Equation (\ref{Heis_F}) is equivalent to a first-order partial differential equation, specifically:
\[
\partial_t f = (q+\lambda E_{n,k,s, \theta})\partial_p f -p\partial_q f.
\]
Applying the method of characteristics we find that $f$ is a constant along each characteristic curve. Each characteristic curve in the $(q,p)$ plane follows a different trajectory based on the values of $E_{n,k,s,\theta}$. In other words, $f(t,q,p)=f(0, q(-t), p(-t))$ where
\[
\begin{bmatrix}
q(t) \\ p(t)
\end{bmatrix}
=
\begin{bmatrix}
\cos(t) & -\sin(t) \\
\sin(t) & \cos(t)
\end{bmatrix}
\begin{bmatrix}
q + \lambda E_{n,k,s,\theta} \\ p
\end{bmatrix}
-
\begin{bmatrix}
\lambda E_{n,k,s,\theta} \\ 0
\end{bmatrix}
\]
This describes the dynamics of the infinite array of qubits explicitly.
The array of qubits may be viewed as an example of a Quantum Metamaterial (QMM). More precisely, when the QMM collapses on a particular quantum state (through a measurement on this subsystem), it shifts the center of rotation within the Wigner-Weyl plane in parallel to the position axis. This, in turn, controls interaction of the field with the QMM. Note that the centers of rotation, corresponding to the values of $E_{n,k,s,\theta}$, densely fill the bounded interval $[-\lambda, \lambda]$ along the $q$-axis.

\begin{figure}[ht!]
\centering
\includegraphics[width=150mm]{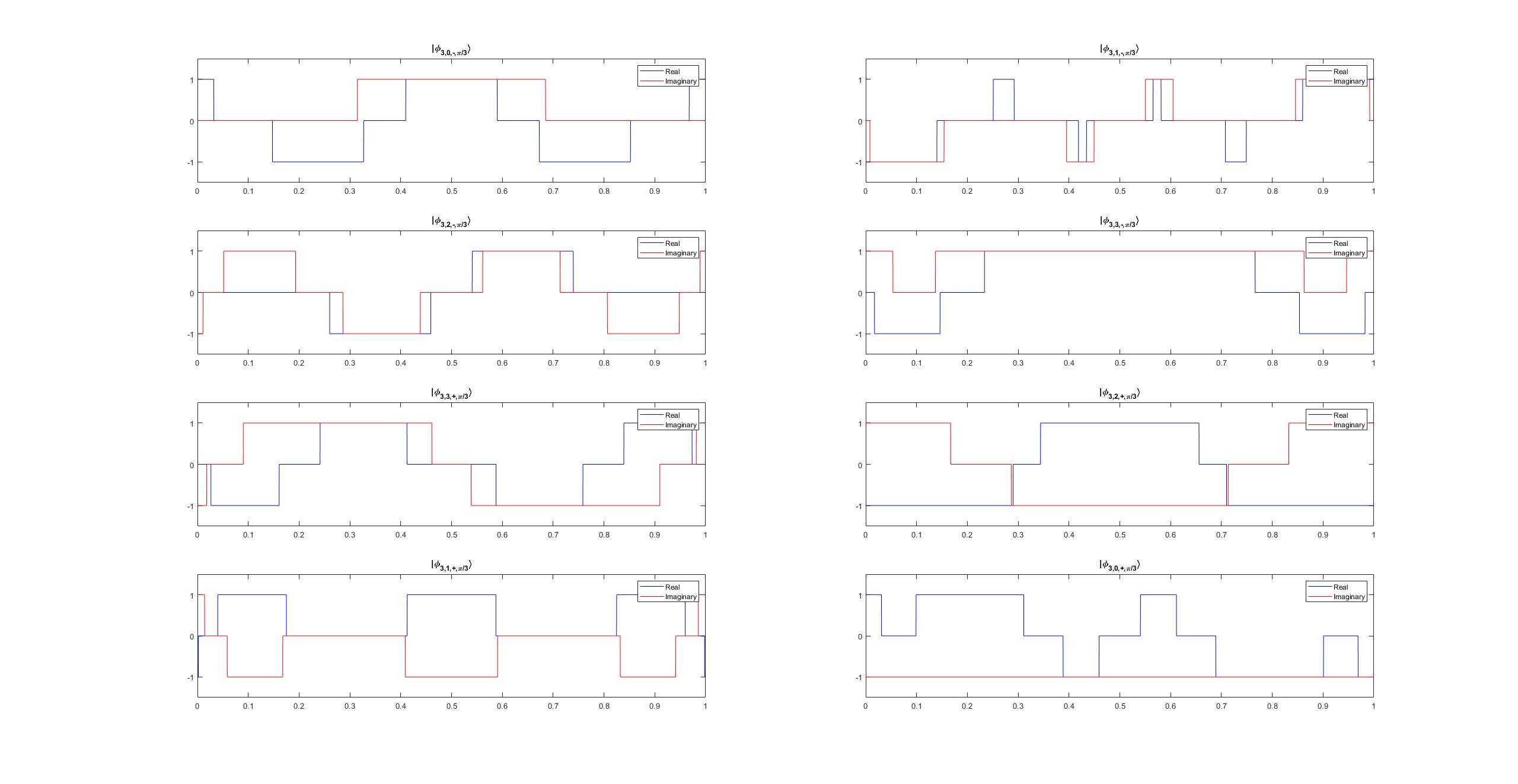}
\caption{Representation of the real and imaginary parts of the eigenstates of operator $C_{\theta}+sin(\theta)K$ in the Haar basis, here $\theta = \frac{\pi}{3}$ and $n = 3$.}
\label{Cthetapi3}
\end{figure}
\newpage
\section*{Summary} 
We have introduced a rigorous framework for analysis of finite or infinite qubit arrays. Our approach leads to representation of  observables by means of geometric operators. We have explored the properties of particular examples of such operators, including nonlocal operators that encode nontrivial action on an infinite array. We have demonstrated that for a special choice of parameters (dyadic scaling), the spectral properties of some essential operators of that type can be described explicitly. We have also applied these results to model a Quantum Metamaterial consisting of an array of qubits.     

Extensive connections have been established between the wavelet transform theory and the realm of quantum physics. In essence, wavelets can be perceived as generalized coherent states, a notion that is comprehensively discussed and reviewed in \cite{Ali}. More to the point, wavelets frequently emerge as the preferred methodology for conducting numerical analyses of quantum mechanical challenges. This preference is well-demonstrated by their prominent utilization in this context, as exemplified in \cite{Griebel, Flad}.

\section*{Acknowledgments}

The authors are immensely grateful to Professor Alexandre Zagoskin for inspiring conversations. 


\newpage
\section{Appendix: heuristic arguments and numerical experiments} \label{section_appendix}

The meaning of statement (\ref{blocks_C}), proven in \cite{R1}, is that $C_x$ has a unitarily equivalent representation in the form of an infinite block matrix. Moreover, the unitary equivalence is given by the Haar transform. Therefore, at a first glance, it would seem to be a good working hypothesis that a similar statement should be true with regards to $C_y$. Also, since operators $\sigma_x$ and $\sigma_y$ are unitarily equivalent one might suspect that this equivalence would somehow transfer over to their periodizations. However, the actual situation turns out to be more complicated. In particular, the two ad hoc suppositions we have just mentioned are not true. 

In this section, we will outline some numerical experiments we have conducted that elucidate this point. The numerical representation of operators requires that they be truncated, by which we mean restricting them to the finite-dimensional subspace $V_n$. In other words, one considers the action of an operator only on the first $n$ qubits in the array. To avoid confusion, we will refer to the truncation of $C_y$ as $C_y^{(n)}$, i.e., $C_y^{(n)} = C_y|_{V_n}$.  We will extend this convention to other operators as well. In all figures we chose $n = 10$. The truncated operators can be represented numerically and investigated via computation up to machine precision. Here are some crucial numerical observations:
  
\begin{itemize}
  \item Representing $C_y^{(n)}$ in the Haar basis does not lead to a block matrix.    
  \item One can define a sequence of blocks in analogy to (\ref{block_D}) and construct an operator which has a block structure in the Haar basis\footnote{The blocks in question are described in Theorem \ref{P+theor} and Corollary \ref{cor_theta}.}. However, this operator will differ from the original $C_y^{(n)}$. The difference between these two operators is the ``remainder" matrix illustrated in Fig. \ref{compact1}.  
\end{itemize}
Note that the entries of the ``remainder" matrix are concentrated near the corner with index $(1,1)$ and diminish very rapidly farther away from it. This suggests that the ``remainder" matrix is a truncation of a compact operator (i.e., an operator that can be effectively approximated by a finite matrix). That is a step forward, but an additional spark of inspiration is required to characterize this operator explicitly and rigorously.

The path to a characterization of the ``remainder" is simplified if one initially focuses on $C_{-}^{(n)}$. Experimentation highlights the role of a matrix, denoted $L^{(n)}$, with entries
\begin{equation}\label{L_ij}
L_{i,j}^{(n)} =
\left\{
    \begin{array}{lr}
        1, & \text{if } i < j        \\
        0, & \text{if } i\geq j   \\
    \end{array}\right., \quad i, j = 1, 2, 3, \ldots 2^n.
\end{equation}
A numerical experiment shows that  
$C_-^{(n)} + L^{(n)}$ exhibits block structure when represented in the Haar basis. This is evidenced in Fig. \eqref{Ps}. It is now easy to develop the right conjecture as to the nature of the ``remainder"; indeed, it is an integral operator in $ L_2(0,1]$ with the kernel \eqref{ell_kernel}. 

Additionally, relation (\ref{EqPy}) determines the right ``remainder" to effect block structure of $C_y^{(n)}$ in the Haar basis, see Fig. \ref{Py}.


\begin{figure}[ht!]
\begin{center}
\includegraphics[width=100mm]{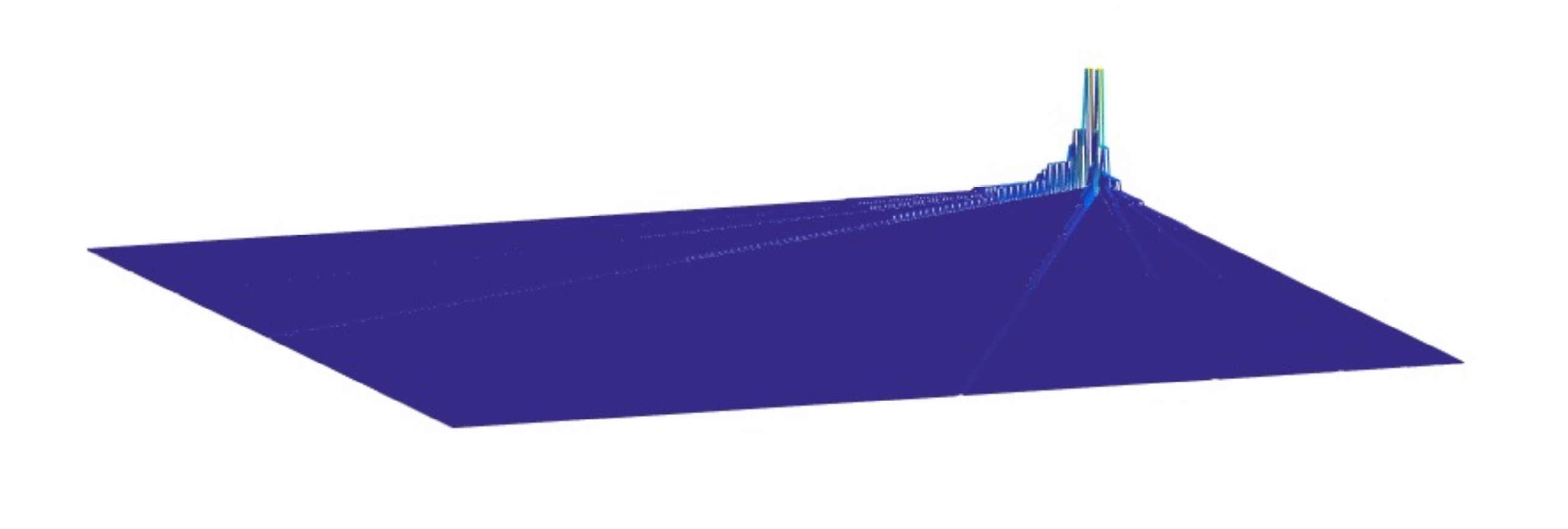}
\caption{The graph of the numerically identified ``remainder"  matrix. The matrix is obtained by subtracting from $C_y^{(n)}$ the part which reduces to blocks under the Haar transform; it is then represented in the Haar basis. The observed accumulation of the significant entries close to the position $(1,1)$ (the corner farthest away from the viewer) suggests that the matrix is a truncation of a compact operator. Theoretical analysis proves that this is indeed the case. However, this graph by itself offers no clue as to how to describe this operator rigorously.}
\label{compact1}
\end{center}
\end{figure}

\begin{figure}[ht!]
\begin{subfigure}{0.5\textwidth}
\includegraphics[width=1\linewidth, height=7cm]{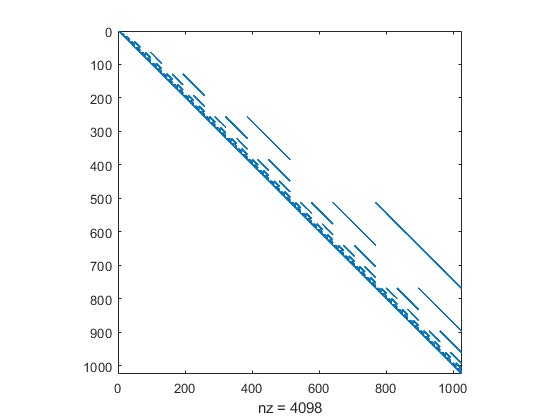}
\caption{non-zero elements of the operator $C_{-}^{(n)}+L^{(n)}$}
\label{P-}
\end{subfigure}
\begin{subfigure}{0.5\textwidth}
\includegraphics[width=1\linewidth, height=7cm]{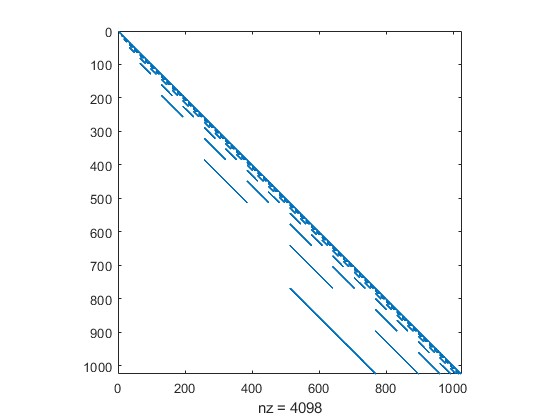} 
\caption{non-zero elements of the operator $C_{+}^{(n)}+(L^{(n)})'$}
\label{P+}
\end{subfigure}
\caption{Location of nonzero entries of truncated operators $C_{-}+L$  and $C_{+}+L'$ (where $L'$ is the adjoint of $L$) when they are represented in the Haar basis. Note the block structure of these matrices.}
\label{Ps}
\end{figure}


\begin{figure}[ht!]
\begin{center}
\includegraphics[width=150mm]{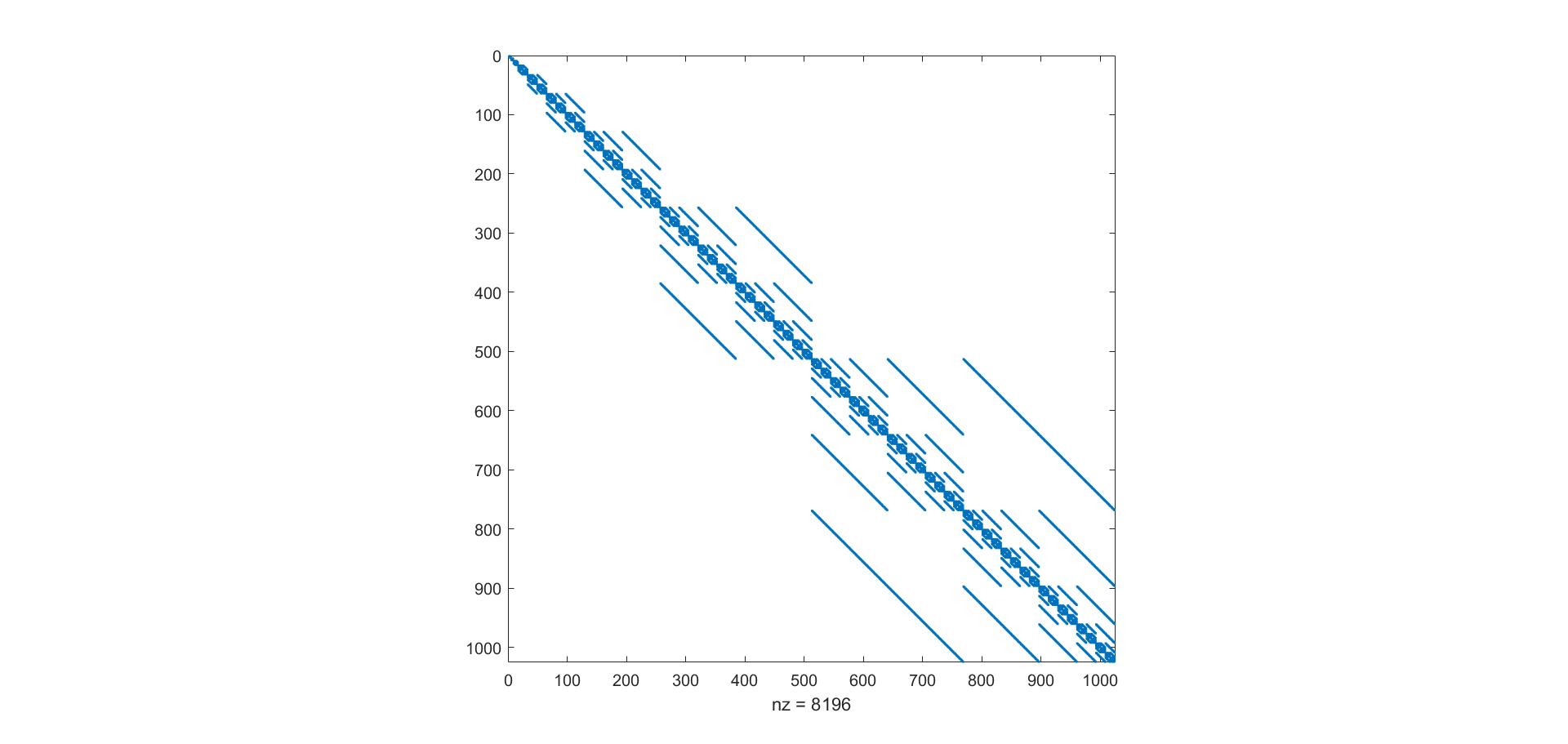}
\caption{Non-zero elements of the operator $C_y^{(n)} + i[(L^{(n)})'-L^{(n)}]$ when represented in the Haar basis. Note the block structure of this matrix. This points at $i[L'-L]$, see definition (\ref{ExplicitL}), as the  right hypothesis for the compact correction that effects block structure.} 
\label{Py}
\end{center}
\end{figure}

\end{document}